\documentclass[12pt]{article}
\usepackage[utf8]{inputenc}

\title{On the Conjugacy Search Problem \\ in Extraspecial $p$-Groups}

\date{}

\author{Simran Tinani\thanks{This research is supported by armasuisse Science and Technology.}}

\usepackage{graphicx}
\usepackage{amsmath, amssymb}
\usepackage[utf8]{inputenc}
\usepackage{subcaption}
\usepackage{tabularx}
\usepackage{amsmath,amssymb,amsthm}
\usepackage[english]{babel}
\usepackage{mathtools}
\usepackage{algorithmic}
\usepackage{enumerate}
\usepackage{hyperref}
\usepackage[numbers]{natbib}
\usepackage[dvipsnames]{xcolor}
\usepackage{color}
\usepackage{comment}

\usepackage{hyperref}
\hypersetup{
  colorlinks   = true, 
  urlcolor     = blue, 
  linkcolor    = OliveGreen, 
  citecolor   = red 
}

\usepackage{tabularx}
\usepackage{booktabs}
\usepackage{makecell}
\usepackage{longtable}
\usepackage[ruled,vlined]{algorithm2e}
\usepackage{algorithmic}
\usepackage{listings}
\usepackage{etoolbox}
\usepackage{lipsum}
\usepackage[left =2.6cm, right = 2.6 cm]{geometry}

\usepackage{url}

\usepackage{amssymb}
\usepackage{amsfonts}
\usepackage{amsmath}
\usepackage{comment}
\usepackage{tabularx}
\usepackage{booktabs}
\usepackage{makecell}
\usepackage{longtable}
\usepackage{graphicx}
\usepackage{caption}
\usepackage{lipsum}
\usepackage{xcolor}
\usepackage{parskip}
\usepackage{caption}
\usepackage{floatrow}

\usepackage{hyperref, enumitem}

\newcommand{\Z}{\mathbb{Z}}

\DeclarePairedDelimiter\card{\lvert}{\rvert}
\newfloatcommand{capbtabbox}{table}[][\FBwidth]

\newtheorem{theorem}{Theorem}

\newtheorem{lemma}{Lemma}

\theoremstyle{definition}

\newtheorem{definition}{Definition}
\theoremstyle{remark}
\newtheorem{remark}{Remark}
\newtheorem{example}{Example}

\begin{document}
\maketitle

\begin{abstract}
     In the recently emerging field of group-based cryptography, the Conjugacy Search Problem (CSP) has gained traction as a non-commutative replacement of the Discrete Log Problem (DLP). The problem of finding a secure class of nonabelian groups for use as platforms is open and a subject of active research. This paper demonstrates a polynomial time solution of the CSP in an important class of nonabelian groups, the extraspecial $p$-groups. For this purpose, and as a further result, we provide a reduction of the CSP in certain types of central products. The consequences of our results are practically relevant for ruling out several groups as platforms, since several nonabelian groups are constructed by combining smaller groups by taking direct and central products.
\end{abstract}

\section{Introduction}

Among the search for a post-quantum framework for public key cryptography, it  has been proposed recently to use algorithmic problems over onabelian groups. We recall the Discrete Logarithm Problem (DLP), which forms the basis for a large majority of classical key exchange, authentication, and signature schemes:  Given a group $G$, an integer $N$, and group elements $g$ and $h=g^N$, find an exponent $n \in \Z$ such that $h=g^n$. On the other hand, the Conjugacy Search Problem (CSP) is described as follows. Given a group $G$ and elements $g, \ y$ and $h=y^{-1} g y$, the CSP requires the recovery of a conjugator $x \in G$ such that $h = x^{-1} gx$. The notation $g^x:= x^{-1}g x$ is often used to reflect this analogy. For background on group-based cryptography, the interested reader can refer to \cite{myasnikov2008group} and \cite{fine2011aspects}. 
 
The most well-known protocols constructed based on the CSP are by Anshel, Anshel and Goldfeld (AAG) \cite{an99}, and Ko-Lee \cite{ko2000new}. Both of these protocols propose the use of the Braid groups $B_N$ as platforms for key exchange. However, the braid groups have been shown to be vulnerable to a number of attacks \cite{braid1}, \cite{braid2}, \cite{tsaban2015polynomial} and are thus likely unsuitable. The problem of finding a secure class of groups for use as platforms for CSP-based protocols is open and a subject of active research. 

A finite group with order a power of a prime $p$ is called a $p$-group. The class of $p$-groups is vast, and not fully understood in its entirety. Among these, some interesting and well-studied subclasses are the special, extraspecial, and Miller $p$-groups. These classes of groups are overlapping; in fact, every extraspecial $p$-group is also special, and most well-known Miller groups are also special. For a general reference on $p$-groups, see \cite{LeedhamGreen2002TheSO}. 

Since $p$-groups constitute important classes of nonabelian groups, and often form building blocks for other nonabelian groups, it is worth examining the difficulty of the CSP in them. In fact, some authors have already proposed them as potential platforms for cryptography. For example, in \cite{new-auth}, authentication and signature schemes using the CSP were proposed, and general $p$-Miller groups were suggested as platforms.

A $p$-group $G$ is called extraspecial if its center $Z(G)$ is cyclic of order $p$, and the quotient $G/Z(G)$ is a non-trivial elementary abelian $p$-group. Every extraspecial $p$-group has order $p^{1+2n}$ and is a central product of $n$ extraspecial groups of order $p^3$. 

In this paper, we show that the CSP in an extraspecial $p$-group has a polynomial time solution, and provide an explicit algorithm, Algorithm \ref{csp-algo} for the solution. For this, we first solve the CSP in extraspecial groups of order $p^3$, and show an extension to all extraspecial $p$-groups (section \ref{extraspl}). Such an extension requires exploring CSPs in central products of groups, which we do in section \ref{central-sec}. More precisely, we will show in this section that an instance of the CSP in a central product $G=HK$ reduces to two separate instances of the CSP in $H$ and $K$, if $G$ satisfies a certain algorithmic condition, which we call efficient $C$-decomposability. We also provide examples of groups which possess this property, and show later on, in section \ref{extraspl} that it is possessed by all extraspecial $p$-groups.

 A direct consequence of the results in this paper is that nonabelian groups with extraspecial $p$-groups as direct or central components are unsuitable as cryptographic platforms. The results on central products also demonstrate that while considering any platform group for a CSP-based system, care must be taken to ensure that an efficient decomposition into a central product is not possible. This is practically significant for future work in group-based cryptography since several nonabelian groups, and in particular, $p$-groups, are constructed by combining smaller $p$-groups by taking direct, central, and semidirect products (see, for example, \cite{caranti2013}, \cite{curran}).

As a direct corollary of our results, we also have a polynomial time solution for the Conjugacy Decision Problem (CDP) (i.e. given elements $g,h \in G$, to determine if there exists $x\in G$ such that $x^{-1}gx = h$) in any extraspecial $p$-group.

Throughout, we denote by $p$ a prime number, $C_p$ the cyclic group of order $p$, and by $H \rtimes K$ a semidirect product of groups $H$ and $K$ (with $K$ acting on $H$ by automorphisms). The center of a group $G$ is denoted by $Z(G)$. For a subset $S$ and an element $x$ of a group $G$ we use the notations $xS := \{xz \mid z\in S\}$ and $Sx := \{zx \mid z\in S\}$. For any $x\in G$, denote by $C_x:=\{g^{-1}xg \mid g \in G\}$ the conjugacy class of $x$. For subsets $S_1$ and $S_2$, the product $S_1S_2$ denotes the set $\{s_1s_2 \mid s_1\in S_1, \ s_2 \in S_2\}$. An algorithm on a finite group $G$ will be said to be polynomial time if its time complexity is $\mathcal{O}(\log \card{G})$. 

\section{Central Products}\label{central-sec}

\begin{definition}
A group $G$ is said to be a central product of its subgroups $H$ and $K$ if every element $g \in G$ can be written as $hk$, with $h\in H, k \in K$ (i.e. $G=HK$), and we have $h k = kh \ \forall \ h\in H, \ k\in K$ (thus, $H \cap K \subseteq Z(G))$.
\end{definition}

We introduce the following property, which is relevant to central products and the CSP. 

\begin{definition}
A finite group $G$ is said to be efficiently $C$-decomposable if for any elements $h,  k, x, y \in G$ with $hC_x \cap kC_y \neq \emptyset$, an element of $hC_x \cap kC_y $ can be found in polynomial time. 
\end{definition}

\begin{theorem}
\label{central-thm}
     Let $G$ be an efficiently $C$-decomposable group and $H$ and $K$ be subgroups of $G$ such that $G$ is the central product of $H$ and $K$. Then, solving the CSP in $G$ is polynomial time reducible to solving two separate CSP's in $H$ and $K$. 
\end{theorem}
\begin{proof}
     Let $\Tilde{g} = \Tilde{h}\Tilde{k}$, $g' = h'k'$ be elements in $G$. Suppose that we want to solve the CSP for $\Tilde{g}$ and $g'$, i.e. to find an element $=hk$ such that $g^{-1} \Tilde{g} g = g'$. We have,
     \begin{align}\label{conj_central}
         & (hk)^{-1}(\Tilde{h}\Tilde{k})(hk) = k^{-1}h^{-1} (\Tilde{h}\Tilde{k}) hk  =(h^{-1}\Tilde{h} h)(k^{-1}\Tilde{k} k) \nonumber \\
         \implies  & (h^{-1}\Tilde{h} h)(k^{-1}\Tilde{k} k) = h'k' \nonumber \\
          \implies & {h'}^{-1}(h^{-1}\Tilde{h} h) = k'(k^{-1}\Tilde{k} k)^{-1} \in h'^{-1}C_{\Tilde{h}} \cap k'C_{\Tilde{k}^{-1}} \subseteq H\cap K
     \end{align}
   Note that ${h'}^{-1}(h^{-1}\Tilde{h} h) =k'(k^{-1}\Tilde{k} k)^{-1} \in {h'}^{-1}C_{\Tilde{h}} \cap k'C_{\Tilde{k}} \subseteq H \cap K$. By hypothesis, we can find, in polynomial time, an element $t \in h'^{-1}C_{\Tilde{h}} \cap k'C_{\Tilde{k}^{-1}}$.

  Now consider the following two separate instances of the CSP in $H$ and $K$.
    \begin{equation}\label{central-eqn}
        h^{-1}\Tilde{h} h = h't \in H \ \text{and} \ (k^{-1}\Tilde{k} k) = k' t^{-1} \in K
    \end{equation} 
    Suppose that we have solutions $h$ and $k$ of \eqref{central-eqn}. Then, for $g = hk$, we have $g^{-1} \Tilde{g} g = (hk)^{-1}(\Tilde{h}\Tilde{k})(hk) = (h^{-1}\Tilde{h} h)(k^{-1}\Tilde{k} k) = (h't)(k't^{-1}) = h'k'$. Thus, $g =hk$ is a solution to $g^{-1} \Tilde{g} g = g'$. We conclude that $\Tilde{g}$ and $g'$ are conjugate if and only if $ {h'}^{-1}C_{\Tilde{h}} \cap k'C_{\Tilde{k}} \neq \emptyset$, and that in this case, a conjugator can be found by solving \eqref{central-eqn}.
\end{proof}

The following lemma is a well-known number theoretic fact which lies at the root of several results of this paper.

\begin{lemma}
     The modular linear equation $ax \equiv b \mod n$ in $x$ has a solution if and only if $\gcd(a,n) \mid b$. When a solution exists, it may be found in time $\mathcal{O}(\log n)$.
\end{lemma}
\begin{example}
The dihedral groups $D_n$ are all efficiently $C$-decomposable. Consider $D_n$, which has the presentation $\langle x,y \mid x^n=1, y^2=1, yx = x^{-1}y\rangle$. Then we have $y^j x^i = x^{i(-1)^j}y^j$ for all $i, j$,  
\[C_{x^Ay^B} = \{x^{i+A(-1)^j-i(-1)^B}y^B \mid 0\leq i \leq n-1, \ 0\leq j\leq 1\}\] 

\[(x^ky^l)C_{x^Ay^B} = \{x^{k
+(-1)^l[i+A(-1)^j-i(-1)^B]}y^{l+B} \mid 0\leq i \leq n-1, \ 0\leq j\leq 1\}\]

Write $h= x^{k_1}y^{l_1}, u = x^{A_1}y^{B_1}, k=x^{k_2}y^{l_2}, v=x^{A_2}y^{B_2}$. It is then easily verified that finding a common element of $hC_u$ and $kC_v$ amounts to finding $(i_1, i_2, j_1, j_2)$ from an equation of the form
\[k_1
+(-1)^{l_1}[i_1+A_1(-1)^{j_1}-i_1(-1)^{B_1}] = k_2
+(-1)^{l_2}[i_2+A_2(-1)^{j_2}-i_2(-1)^{B_2}] \mod n\]

  Clearly, one may choose values of $j_1$ and $j_2$ and then find a linear relation between $i_1$ and $i_2$ by solving a linear modular equation. Thus, a solution can be found in time $\mathcal{O}(\log n)$.
  \end{example}

\begin{example}
The generalized quaternion groups $Q_{2^n}$ are all efficiently $C$-decomposable. A generalized quaternion group is given by the presentation
\begin{equation}\label{quatpres}
    Q_{2^n} = \langle x, y \mid x^N = 1, y^2 = x^{N/2}, yx = x^{-1}y, N= 2^{n-1}\rangle. 
\end{equation}
We derive the relations $x^i y = yx^{N-i}$ and $yx^i= x^{N-i} y$. Using these one easily derives the relation ${y^j}{x^i} =  x^{i(-1)^j}y^j$ for all $i,j \in \Z$.
\[(x^ky^l) C_{x^Ay^B} = \{x^{k+(-1)^{j+l}[A-i+i(-1)^B]}y^{B+l} \mid 0\leq i\leq N, \ 0\leq j \leq 1\}\]

Write $h= x^{k_1}y^{l_1}, u = x^{A_1}y^{B_1}, k=x^{k_2}y^{l_2}, v=x^{A_2}y^{B_2}$. Then finding a common element of $hC_u$ and $kC_v$ amounts to solving an equation of one of the following forms for $(i_1, i_2, j_1, j_2)$:
\[k_1
+(-1)^{j_1+l_1}[A_1-i_1+i_1(-1)^{B_1}] = k_2
+(-1)^{j_2+l_2}[A_2-i_2+i_2(-1)^{B_2}] \mod N\]

  Again, one may choose values of $j_1$ and $j_2$ and then find a linear relation between $i_1$ and $i_2$ by solving a linear modular equation. Thus, a solution can be found in time $\mathcal{O}(\log N)$.

\end{example}

\begin{remark}
Consider the group $S_n$ of permutations on $n$ elements. If $x$ and $y \in S_n$ are cycles, then finding an element of $hC_x\cap kC_y$ for any $h,k \in S_n$ can be done in polynomial time. This follows from theorem 5 in \cite{HERZOG2004323}, by taking $\sigma = h^{-1}k$ and noting that the result used at each step is in fact constructive and achievable in polynomial time. However, for $x$ and $y$ general permutations, a procedure to find an element in $hC_x\cap kC_y$ is not, in the author's knowledge, so far clear. 
\end{remark}

\section{Extraspecial $p$-groups}\label{extraspl}

\begin{definition}[Extraspecial $p$-group]
A $p$-group $G$ is called extraspecial if its center $Z(G)$ is cyclic of order $p$, and the quotient $G/Z(G)$ is a non-trivial elementary (i.e. every element has order $p$) abelian $p$-group.
\end{definition}

 The following results are standard, and can be found, for instance, in \cite{gorenstein2007finite}.

\begin{theorem}
There are precisely two isomorphism classes for the extra special group of order $p^3$: $M(p)= C_{p^2} \rtimes C_p$ and $N(p)=(C_p \times C_p ) \rtimes C_p$, where the latter may be represented as triangular matrices over the finite field of order $p$, with 1's on the diagonal. 
\end{theorem}

\begin{theorem} 
Every extraspecial $p$-group has order $p^{1+2n}$ for some positive integer $n$, and conversely for each such number there are exactly two extraspecial groups up to isomorphism. Every extraspecial group of order $p^{1+2n}$ can be written as a central product of either $n$ copies of $M(p)$ or $n-1$ copies of $M(p)$ and 1 copy of $N(p)$. 
\end{theorem}

A polynomial time algorithm for the computation of such a central product decomposition was given in \cite{Wilson+2009+813+830}. 

\begin{theorem}[\cite{Wilson+2009+813+830}]\label{central-p3}
Let $P$ be a finite $p$-group of class 2 (i.e. $P'\leq Z(P)$) for a known prime $p$. Then, there is a Las-Vegas polynomial-time (in the order of $G$) algorithm which returns a set $\mathcal{H}$ of subgroups of $P$ (called a central decomposition of $P$) where distinct members pairwise commute, with $\mathcal{H}$ of maximum size such that $\mathcal{H}$ generates $P$ and no proper subset does.
\end{theorem}

The algorithm referred in the above theorem assumes that the group $P$ is given by its generators, $P=\langle S \rangle$. A particular representation is not selected, but it is assumed that $P$ can be input with $\mathcal{O}(\card{S}n)$ bits of data and that there are polynomial-time algorithms for multiplication, inversion, equality testing, and subgroup membership testing in $P$. 

\begin{remark}
We observe that given the generators of an extraspecial group of order $p^3$, one can find, in polynomial time, which of the two isomorphism classes this group belongs to, simply by raising each generator to the power $p$.
\end{remark}

\subsection{CSP in Extraspecial Groups of order $p^3$}\label{pgps}

In the group $M(p)=C_{p^2} \rtimes C_p$, elements are clearly represented as words $x^iy^j$ with $x \in C_{p^2}$ and $y \in C_p$, whereas in $N(p)=(C_p \times C_p ) \rtimes C_p$, they may be represented as matrices, but also equally well in the semidirect notation, as it is easy to switch between both the presentations. We deal with these two groups separately.

\subsubsection{$M(p)= C_{p^2} \rtimes C_p$}

It is well-known (refer, for example, to \cite{conrad2014groups}) that $M(p)$ has the presentation \begin{equation}\label{extra1}
    M(p) =\langle x, y \mid x^{p^2}=1, y^p=1, yxy^{-1} = x^{1+p} \rangle
\end{equation} 

The following lemma can be easily verified by induction.

\begin{lemma}\label{rel-extra1}
     The following relation holds in $M(p)$. 
\begin{align*}
    y^j x^i  & = x^{{i(1+p)^j}}y^j =  x^{{i(1+jp)}}y^j\ \forall \ i,j \in \Z
\end{align*}
Here, all powers of $x$ are taken $\mod p^2$ and all powers of $y$ are taken $\mod p$. 
\end{lemma}

\begin{theorem}\label{csp-extra1}
Consider two elements $g = x^a y^b $ and $g' = x^A y^B$ of $M(p)=C_{p^2} \rtimes C_p$. Then, $g$ and $g'$ are conjugates if and only if $a=A \mod p$, $B = b \mod p$. 
In this case, a conjugator $h=x^i y^j$ such that $h^{-1} g h$ can be found by solving $(A - a)/p = (aj-ib) \mod p$. Consequently, the CSP has a polynomial time solution in $M(p)$.
\end{theorem}

\begin{proof}
 We first try to find a solution for $j\geq 0$. We have, by recursively using the relations in Lemma \ref{rel-extra1},
 \begin{equation}\label{extra1-conj}
     (x^iy^j)^{-1}(x^a y^b)(x^iy^j) =x^{(a-i)(1-jp)} y^{b-j} x^i y^j =x^{a-p(aj-bi)}y^b  \ 
 \end{equation}
 
Thus, by the uniqueness of a semidirect product representation,
\begin{align*}
    g^\prime = g^h &\iff b=B \mod p, \ \text{and} \ A= a{(1-jp)}+ibp \mod p^2 
\end{align*}
In particular, $A = a \mod p$. However, if we do have this condition, then the equation reduces to $(A - a)/p = (ib-aj) \mod p$, which is a linear modular equation in two variables. Clearly, one may fix a value of one of the variables and then solve the other, so that a solution $(i,j)$ satisfying the last condition can be found in polynomial time.
\end{proof}

\subsubsection{$N(p)= (C_{p}\times C_p) \rtimes C_p$}

$N(p)$ has the presentation \begin{equation}\label{extra2-rel}
    N(p) =\langle x, y, z \mid x^p= y^p=z^p=1, xy=yx, yz=zy, zxz^{-1} = xy^{-1} \rangle
\end{equation}
Note that $N(p)$ is also known as the Heisenberg group over $\Z_p$, and can also be seen as the subgroup of $Mat_3(\mathbb{F}_p)$ with 1's along the diagonal, where $x$ corresponds to $X=\left(\begin{smallmatrix}
1 & 1 & 0\\
0 & 1 & 0 \\
0 & 0 & 1
\end{smallmatrix}\right)$, $y$ corresponds to $Y = \left(\begin{smallmatrix}
1 & 0 & 1\\
0 & 1 & 0 \\
0 & 0 & 1
\end{smallmatrix}\right)$, and $z$ corresponds to $Z = \left(\begin{smallmatrix}
1 & 0 & 0\\
0 & 1 & 1 \\
0 & 0 & 1
\end{smallmatrix}\right)$. Further, representations of $G$ in terms of generators and relators, and in terms of matrices, are interchangeable. Given a matrix $G=\left(\begin{smallmatrix}
1 & a & b\\
0 & 1 & c\\
0 & 0 & 1
\end{smallmatrix}\right)$, one can easily verify that $G= X^a Y^b Z^c$. Thus, it is enough to solve the CSP in the generator and relation setting. The following lemma is easily verified by induction.

\begin{lemma}
     The relation $z^ k x^a z^{-k}= x^a y^{-ka}$ holds in $N(p)$ for all $a,\ k \in \Z$. Thus, $z^ k x^a y^b = x^a y^{-kab}z^k$ for all $a,\ k \in \Z$.
\end{lemma}

\begin{theorem}\label{csp-extra2}
The elements $g = x^a y^b z^c $ and $g' = x^A y^B z^C $ in $N(p)$ are conjugate if and only if $a = A \mod p$ and $C=c \mod p$. In this case, $h = x^i y^j z^k$ is a conjugator such that $g' = h^{-1} (x^a y^b z^c)h$ if and only if $(i,k)$ satisfies $B-b = ka -ic$. Consequently, the CSP has  a polynomial time solution in $N(p)$.
\end{theorem}
\begin{proof} Write $g = x^a y^b z^c $, $g' = x^A y^B z^C$, and $ h = x^i y^j z^k$ 
 We have, \begin{align*}
   (x^iy^jz^k)^{-1}(x^a y^b z^c) x^{i}y^{j} z^{k} 
    = & z^{-k} x^{a-i}z^c x^iz^ky^b \\
    = & x^{a-i}y^{k(a-i)}z^{-k}x^i y^{-ci}z^cz^k y^b \\
    =& x^{a-i}y^{k(a-i)}x^i y^{ki}z^{-k}y^{b-ci}z^{k+c} \\
    = & x^{a}y^{ka-ic+b}z^{c}
\end{align*} 
Thus for $g$ and $g'$ to be conjugate, we need $A= a \mod p$, $C = c \mod p$, and that there exist a solution $(i,k)$ to $B-b = ka -ic \mod p$ (note that $j$ can take any value). Given that $g$ and $g'$ are conjugate, if $a = c= 0 \mod p$, then $B=b \mod p$ and any tuple $(i,k)$ works. If one of $a$ and $c$, say $a \neq 0\mod p$ (the case for $c \neq 0 \mod p$ is analogous), then one can choose a value of $i$ and solve $k = (B-b+ic)/a \mod p$ for $k$. Clearly, this is a polynomial time solution.
\end{proof}

\subsection{Lifting the solution to central products }
We now prove efficient $C$-decomposability for extraspecial $p$-groups. This will then allow us to lift the individual solutions of the CSP in $M(p)$ and $N(p)$ to all extraspecial $p$-groups.

\begin{theorem}\label{desc-class}
 Any central product $G$ of finitely many copies of $N(p)$ and $M(p)$ is efficiently $C$-decomposable.
\end{theorem}
\begin{proof}
 We first deal with the case $G=M(p)$.
 More specifically, write $g' =x^a y^b$ and $g = x^c y^d$. Then, from the proof of Theorem \ref{csp-extra1} we have
\begin{align*}
C_{g'} &= \{x^{a-p(aj-bi)}y^b \mid i,j \in \Z,\} \ \text{so} \\
    gC_{g'}& = \{(x^c y^d)(x^{a-p(aj-bi)} y^b) \mid i,j \in \Z\} \\
 & = \{ x^{c+a-p(aj-bi-ad)} y^{d+b}\} \mid i,j \in \Z\}
\end{align*}

Writing $g_i = x^{c_i}y^{d_i}$, $g_i' = x^{a_1}y^{b_1}$, $i=1,\ 2$. It is clear that if an element common to the sets $g_1C_{g_1'}$ and $g_2C_{g_2'}$ exists, then we must have $b_1+d_1 = b_2+d_2 \mod p$, $c_1+a_1 = c_2+a_2 \mod p$ and such an element may be found by solving
\begin{align*}
    \frac{1}{p}[(c_1+a_1)-(c_2+a_2)]=(a_1j_1-b_1i_1-a_1d_1)-(a_2j_2-b_2i_2-a_2d_2) \mod p
\end{align*}
for $i_1, i_2, j_1, j_2$. Clearly, this is a linear equation, so a solution can be found in time $\mathcal{O}(\log p)$.

Now, for $G= N(p)$, we write $g' =x^a y^b z^c$ and $g = x^d y^e z^f$. Then, from the proof of Theorem \ref{csp-extra2} we have
\begin{align*}
C_{g'} &= \{x^{a}y^{ka-ic+b}z^c \mid i,k \in \Z,\} \ \text{so} \\
    gC_{g'}& = \{(x^d y^e z^f)(x^{a}y^{ka-ic+b}z^c) \mid i,k \in \Z\} \\
    &= \{(x^{d+a} y^{e-fa+ka-ic+b}z^{f+c}) \mid i,k \in \Z\}
\end{align*}

Writing $g_i = x^{d_i}y^{e_i}z^{f^i}$, $g_i' = x^{a_i}y^{b_i}z^{c^i}$, $i=1,\ 2$, It is clear that if an element common to the sets $g_1C_{g_1'}$ and $g_2C_{g_2'}$ exists, then we must have $a_1+d_1 = a_2+d_2 \mod p$, $c_1+f_1 = c_2+f_2 \mod p$ and such an element may be found by finding $(i_1,j_1, i_2, j_2)$ satisfying
\begin{align*}
    \frac{1}{p}[(c_1+a_1)-(c_2+a_2)]=(a_1j_1-b_1i_1-a_1d_1)-(a_2j_2-b_2i_2-a_2d_2) \mod p
\end{align*}

Thus, $M(p)$ and $N(p)$ are both efficiently $C$-decomposable. 

For the statement on central products, first observe from that for a central product $G= H_1 \ldots H_r$, the conjugacy class $C^G_{g} $ of $g=h_1 \ldots h_r$ in $G$ can be written as the product $C^G_{g} = C_{h_1}\ldots C_{h_r}$ of conjugacy classes $C_{h_i}$ of $h_i$ in $H_i$. Further, 
for any $g' = h_1'\ldots h_r' \in G$, $g'C_g = h_1'C_{h_1} \ldots h_r'C_{h_r}$. 
Without loss of generality $H_i=M(p), 1\leq i \leq r$, $H_i=N(p), r+1\leq i \leq s+r$, so $x_1^p=x_2^p=\ldots =x_r^p = y_{r+1} = \ldots = y_{r+s}$ (in this case the nontrivial central product is uniquely defined by this condition).

Write $g = h_1\ldots h_{s+r}$, $g' = h_1'\ldots h_{s+r}'$, $h_i, h_i' \in H_i$. By the discussion above, there exist linear polynomials $A_i(s,t)$ and constants $B_i$ $1 \leq i\leq r$ and linear polynomials  $B_i(u)$ and constants $A_i, C_i$, $r+1 \leq i\leq r+s$ each determined by $h_i, h_i'$, such that \begin{align}\label{cp-expr}
    g'C_g = & h_1'C_{h_1} \ldots h_r'C_{h_r} \nonumber \\
    =& (x_1^{A_1(s_1,t_1)}y_1^{B_1})\ldots (x_r^{A_r(s_r,t_r)}y_r^{B_r} )(x_{r+1}^{A_{r+1}}y_{r+1}^{B_{r+1}(u)}z_{r+1}^{C_{r+1}}) \ldots (x_{r+s}^{A_{r+s}}y_{r+s}^{B_{r+s}(u)}z_{r+s}^{C_{r+s}}) \nonumber \\
    =& (x_1^{A_1(s_1,t_1)} \ldots x_r^{A_r(s_r,t_r)} x_{r+1}^{A_{r+1}}\ldots x_{r+s}^{A_{r+s}}) (y_1^{B_1}\ldots y_r^{B_r}y_{r+1}^{B_{r+1}(u)} \ldots y_{r+s}^{B_{r+s}(u)}) (z_{r+1}^{C_{r+1}} \ldots z_{r+s}^{C_{r+s}})
\end{align}
  Note that this expression is not unique, as the intersections $H_i \cap H_j, \ i \neq j$ are not empty. One may further reduce the $A_i(s_i,t_i)$'s by substituting  $x_1^p=x_2^p=\ldots =x_r^p = y_{r+1} = \ldots = y_{r+s}$, so without loss of generality we may assume that in the above expression, all exponents are less than $p$. Under this restriction, such an expression is unique. Now in order to find an element common to $g_1C_{g_1'}$ and $g_2C_{g_2'}$  one may equate the polynomials in the exponents of the $x_i$'s and $y_i$'s (note that all the exponents of the $z_i$'s are constants so they must be equal if the intersection is nonempty) and solving for the integers $s_i, t_i, u_i$ as done above, individually in each group. This requires the solution of at most $r+s$ linear equations $\mod p$, and thus has complexity $\mathcal{O}((r+s)\log p) = \mathcal{O}(\log |G|)$. Thus, a polynomial time solution exists, and $G$ is efficiently $C$-decomposable.
\end{proof}

\begin{theorem}
The CSP in an extraspecial $p$-group has a polynomial time solution.
\end{theorem}
\begin{proof}
 Let $G$ be an extra special $p$-group. By Theorem \ref{central-p3}, a decomposition of $G$ into finitely many (say, $r$-many) extraspecial groups of order $p^3$-can be done in polynomial time. By Theorem \ref{desc-class}, $G$ is efficiently $C$-decomposable. Applying Theorem \ref{central-thm}, the CSP in $G$ reduces in polynomial time to $r$ CSP's in extraspecial $p^3$ groups, which by Theorems \ref{csp-extra1} and \ref{csp-extra2}, have polynomial time solutions. In Algorithm \ref{csp-algo}, we outline the step-by-step procedure explicitly. \end{proof}

\begin{algorithm}[ht]  
\hspace*{\algorithmicindent} \textbf{Input} {Generators of an extraspecial $p$-group $G$, conjugate elements $\Tilde{g}, g' \in G$} \\
\hspace*{\algorithmicindent}
\textbf{Output} {A conjugator $g\in G$ such that $g^{-1}\Tilde{g}g=g'$}
\begin{algorithmic}[1]
\STATE Use the algorithm of \cite{Wilson+2009+813+830} to obtain a central product decomposition of $G$ into finitely many extraspecial groups of order $p^3$. Without loss of generality $G= H_1 \ldots H_{r+s}$, with $H_i=M(p), 1\leq i \leq r$, $H_i=N(p), r+1\leq i \leq s+r$.
\STATE Write $g=h_1 \ldots h_{r+s}$, set $j\leftarrow 1$.
\STATE  \textbf{while $j\leq r+s$}
\begin{enumerate}
\item $\bar{K}\leftarrow H_{j+1} \ldots H_{r+s}$,  $\bar{H}\leftarrow H_j$. Write $g'=h'k',\Tilde{g}=\Tilde{h}\Tilde{k'}, h', h\in \bar{H}, k',k\in \bar{K}$.
\item Find an element $t_j \in h'^{-1}C_{\Tilde{h}} \cap k'C_{\Tilde{k}} $ by equating expressions of the form \eqref{cp-expr}. 
\item Solve for $h_j$ by solving the CSP $h_j^{-1}h_j'h_j = h_j't_j$ in $H_j$, using theorems \ref{csp-extra1} and \ref{csp-extra2}.
\item $j \leftarrow j+1$.
\end{enumerate}
 \caption{Solving the CSP in an extraspecial $p$-group } \label{csp-algo}
 \end{algorithmic}
\end{algorithm}

\begin{remark}
By the proofs of Theorems 5 and 6, one easily concludes that the CDP also has a polynomial time solution in each of $M(p)$ and $N(p)$. Further, from the proof of Theorem 1, solving the CDP in a central product $HK$ is equivalent to checking if the intersection $h'^{-1}C_{\Tilde{h}} \cap k'C_{\Tilde{k}^{-1}}$ is empty or not. Again, it is clear from the proof of Theorem 7 that this can be checked in polynomial time for any extraspecial $p$-group. Thus, Algorithm 1 is easily modified to a polynomial time algorithm for the CDP in an extraspecial $p$-group, and we have the following result.
\end{remark}
\begin{theorem}
The Conjugacy Decision Problem has a polynomial time solution in an extraspecial $p$-group.
\end{theorem}

Below, we describe an example of an application of our method for the CSP in the extraspecial group of order $p^5$.

\begin{example}
We consider the central product of $M(p)$ and $N(p)$, which is the extraspecial group of order $p^{5}$. Write \begin{align*}
    \Tilde{g} = \Tilde{h}\Tilde{k}, \
    g' = h'k'
\end{align*}
\begin{enumerate}
    \item We find an element $t \in h'^{-1}C_{\Tilde{h}} \cap k'C_{\Tilde{k}} $. Write $\Tilde{h}= x_1^{a_1}y_1^{b_1}$, $h'^{-1} = x_1^{c_1}y_1^{d_1}$, $\Tilde{k}^{-1}= x_2^{a_2}y_2^{b_2}z_2^{c_2}$ (so $\Tilde{k}=x_2^{-a_2}y_2^{-c_2a_2-b_2}z_2^{-c_2}$), $k' = x_2^{d_2}y_2^{e_2}z_2^{f_2}$. We have \begin{align*}
   h'^{-1}C_{\Tilde{h}} = \{x_1^{c_1+a_1-p(a_1j-b_1i-a_1d_1)}y^{d_1+b_1}\} \\
   k'C_{\Tilde{k}}=\{x_2^{d_2+a_2}y_2^{e_2-f_2a_2 + ka_2-ic_2+b_2}z_2^{f_2+c_2}\}
\end{align*}
Now, we have $x_1^p = y_2$ and $H\cap K = \langle x_1^p \rangle = \langle y_2 \rangle $. 
For the intersection to be nonempty we have \begin{align*}
  c_1+a_1 = 0 \mod p, \quad d_2+a_2 = 0 \mod p \\
d_1+b_1 =0\mod p,  \quad f_2+c_2 = 0 \mod p \\
  a_1j_1-b_1i_1-a_1d_1 =e_2-f_2a_2+k_2a_2-i_2c_2+b_2 \mod p
\end{align*}

The last of these equations has an easy solution: if at least one of $a_1, b_1, a_2, c_2 \neq 0$, we set all the remaining coefficients to zero and solve for the remaining one. Let $i_1$ and $j_1$ be the chosen values. Then $t = x_1^{{c_1+a_1}-p(a_1j-b_1i-a_1d_1)}$. 
\item We now solve the CSP $h^{-1}\Tilde{h} h = h't \in H $. We have \begin{align*}
    h't &= x^{-c_1(1+p)^{-d_1}}y^{-d_1}x_1^{{c_1+a_1}-p(a_1j-b_1i-a_1d_1)} \\ &= x^{-c_1(1+p)^{-d_1}} x^{[c_1+a_1-p(a_1j_1-b_1i_1-a_1d_1)](1+p)^{-d_1}}y^{-d_1}\\ & = x_1^{(1-d_1p)[a_1-p(a_1j_1-b_1i_1-a_1d_1)]}y_1^{-d_1}.
\end{align*} Write $A_1 =(1-d_1p)[a_1-p(a_1j_1-b_1i_1-a_1d_1)]$. Now, $\Tilde{h}$ and $h'$ are conjugates, so $A_1 = a_1 \mod p$, and a conjugator $h$ is given by $h=x_1^{i}y_1^{j}$ where $ib_1-ja_1 = (A_1-a_1)/p= -a_1j_1+b_1i_1$.

\item We similarly solve $ (k^{-1}\Tilde{k} k) = k' t^{-1} = x_2^{d_2}y_2^{e_2}z_2^{f_2} (x_1^{-c_1-a_1+p(a_1j-b_1i-a_1d_1)})$ \\ $=x_1^{-c_1-a_1+p(a_1j-b_1i-a_1d_1+e_2)}x_2^{d_2}z_2^{f_2}=y_2^{-(c_1+a_1)/p +(a_1j_1-b_1i_1-a_1d_1+e_2)}$ for $k$. Recall that $\Tilde{k}=x_2^{-a_2}y_2^{-c_2a_2-b_2}z_2^{-c_2}$. As required we already have $a_2 + d_2 = 0 \mod p, f_2 +c_2 =0 \mod p$. Write $B=-(c_1+a_1)/p +(a_1j_1-b_1i_1-a_1d_1+e_2)$, $b= -c_2a_2-b_2, k=x_2^{I}y_2^Jz_2^{J}$. We then must have $B-b=Ka-Ic$, and $J$ can take any value.
\item The final solution is given by $g=hk$.\end{enumerate}

We now illustrate this situation with the help of a numerical example. Take \begin{align*} p= 29, \ a_1=14, \ c_1=15, \ b_1 = 2, b_2=4, \ d_1=27, \ a_2 = 7, \ d_2 = 22, \ c_2= 6, \ f_2 = 23, e_2=5 \end{align*}

We have $$\Tilde{h}=x_1^{14}y_2^{2}, \ \Tilde{k}=x_1^{22}y_2^{12}z_2^{23}, \  h'=x_1^{797}y_1^2, \ k'=x_2^{22}y_2^{5}z_2^{23}.$$
 Here, $i_1$ and $j_1$ must satisfy $14j_1-2i_1=1 \mod 29$ so we take $(i_1, j_1) = (15,0)$. Then, $t= x_1^{29}$ and $h't = x_1^{826} y_1^2$. 
 
 Now for $h=x_1^i y_1^j$ one requires $a_1j-b_1i=(826-14)/29$, so we may take $(i,j) = (14,0)$. Then, $h=x_1^{14}$ satisfies $h^{-1}\Tilde{h}h=h't$. 
 
 Similarly, we have $k't^{-1} = x_2^{22}y_2^{5}z_2^{23}x_1^{-29} = x_2^{22}y_2^{4}z_2^{23} $. Recall $\Tilde{k} = x_2^{-a_2}y_2^{-c_2a_2-b_2}z_2^{-c_2}= x_2^{22} y_2^{12}z_2^{23}$. Now for $k=x_2^I y_2^J z_2^K$ one requires $Ka_2-Ic_2 = 4-12 = 21 \mod 29 $, so we may take $(I,J, K) = (0,0,26)$. Then, $k=z_2^{26}$ satisfies $k^{-1}\Tilde{k}k=k'$. 
Thus, the final conjugator is $g= hk = x_1^{14}z_2^{26}$: \begin{align*}
    g^{-1}\Tilde{g}g = (h^{-1}\Tilde{h}h)(k^{-1}\Tilde{k}k) = (x_1^{826}y_1^{2})(x_2^{22} y_2^4 z_2^{23}) &= (x_1^{14}y_1^{2})(x_2^{22}y_2^{3}z_2^{23}) \\
    g'=h'k' = (x_1^{797}y_1^{2})(x_2^{22}y_2^5z_2^{23}) &= (x_1^{14}y_1^{2})(x_2^{22}y_2^{3}z_2^{23}).
\end{align*}
\end{example}

\section{Conclusion}

The algorithmic complexity of the CSP in different classes of nonabelian groups is a problem pertinent to the field of group-based cryptography. In this paper, we demonstrated a polynomial time solution for the CSP in an extraspecial $p$-group. For this, we first reduced the CSP in extraspecial groups of order $p^3$ to the solution of a set of linear modular equations. Next, we provided a method for reducing CSPs in certain central products, and used this to extend the result to all extraspecial $p$-groups. A direct conclusion from the results in this paper is that nonabelian groups with extraspecial $p$-groups as direct or central components are unsuitable as cryptographic platforms. Further, central products with the efficient $C$-decomposability property that we introduce must be avoided, and a chosen platform must in a way be ``atomic". This is practically relevant even beyond just $p$-groups, since several nonabelian groups are constructed by combining smaller groups by taking direct, central, and semidirect products (see, for example, \cite{caranti2013}, \cite{curran}). Finally, our results also show a polynomial time solution for the CDP in any extraspecial $p$-group.

\bibliographystyle{plain}
\bibliography{references}

\begin{thebibliography}{10}

\bibitem{an99}
Iris Anshel, Michael Anshel, and Dorian Goldfeld.
\newblock An algebraic method for public-key cryptography.
\newblock {\em Math. Res. Lett.}, 6(3-4):287--291, 1999.

\bibitem{caranti2013}
A.~Caranti.
\newblock A module-theoretic approach to abelian automorphism groups.
\newblock {\em Israel Journal of Mathematics}, 205:235–246, 2015.

\bibitem{conrad2014groups}
Keith Conrad.
\newblock Groups of order $p^3$.
\newblock {\em Expository papers on group theory}, 2014.

\bibitem{curran}
M.~J. Curran.
\newblock Direct products with abelian automorphism groups.
\newblock {\em Communications in Algebra}, 35(1):389--397, 2006.

\bibitem{fine2011aspects}
Benjamin Fine, Maggie Habeeb, Delaram Kahrobaei, and Gerhard Rosenberger.
\newblock Aspects of nonabelian group based cryptography: a survey and open
  problems.
\newblock {\em arXiv preprint arXiv:1103.4093}, 2011.

\bibitem{gorenstein2007finite}
D.~Gorenstein.
\newblock {\em Finite Groups}.
\newblock AMS Chelsea Publishing Series. American Mathematical Society, 2007.

\bibitem{new-auth}
Guangguo Han and Chuangui Ma.
\newblock A new authentication and signature scheme based on the conjugacy
  search problem.
\newblock In {\em 2010 Second International Conference on Networks Security,
  Wireless Communications and Trusted Computing}, volume~2, pages 317--320,
  2010.

\bibitem{HERZOG2004323}
Marcel Herzog, Gil Kaplan, and Arieh Lev.
\newblock Representation of permutations as products of two cycles.
\newblock {\em Discrete Mathematics}, 285(1):323--327, 2004.

\bibitem{braid1}
Dennis Hofheinz and Rainer Steinwandt.
\newblock A practical attack on some braid group based cryptographic
  primitives.
\newblock In {\em Public Key Cryptography --- PKC 2003}, pages 187--198,
  Berlin, Heidelberg, 2002. Springer Berlin Heidelberg.

\bibitem{ko2000new}
Ki~Hyoung Ko, Sang~Jin Lee, Jung~Hee Cheon, Jae~Woo Han, Ju-sung Kang, and
  Choonsik Park.
\newblock New public-key cryptosystem using braid groups.
\newblock In {\em Annual International Cryptology Conference}, pages 166--183.
  Springer, 2000.

\bibitem{LeedhamGreen2002TheSO}
Charles~R. Leedham-Green and Susan Mckay.
\newblock The structure of groups of prime power order.
\newblock 2002.

\bibitem{braid2}
Alexei Myasnikov, Vladimir Shpilrain, and Alexander Ushakov.
\newblock Random subgroups of braid groups: An approach to cryptanalysis of a
  braid group based cryptographic protocol.
\newblock In {\em Public Key Cryptography - PKC 2006}, pages 302--314, Berlin,
  Heidelberg, 2006. Springer Berlin Heidelberg.

\bibitem{myasnikov2008group}
Alexei Myasnikov, Vladimir Shpilrain, and Alexander Ushakov.
\newblock {\em Group-based cryptography}.
\newblock Springer Science \& Business Media, 2008.

\bibitem{tsaban2015polynomial}
Boaz Tsaban.
\newblock Polynomial-time solutions of computational problems in
  noncommutative-algebraic cryptography.
\newblock {\em Journal of Cryptology}, 28(3):601--622, 2015.

\bibitem{Wilson+2009+813+830}
James~B. Wilson.
\newblock Finding central decompositions of p-groups.
\newblock {\em Journal of Group Theory}, 12(6):813--830, 2009.

\end{thebibliography}

\end{document}